\RequirePackage{fix-cm}
\documentclass[smallextended,envcountsect,envcountsame]{svjour3}

\smartqed
\journalname{noname}

\usepackage{amsmath,amssymb}
\usepackage{graphicx}
\usepackage{array}
\usepackage{microtype}
\usepackage[hidelinks]{hyperref}
\newcommand{\LABS}{\mathrm{LABS}}
\newcommand{\SLCP}{\mathrm{SLCP}}

\newcommand{\Ptime}{\mathsf P}
\newcommand{\NP}{\mathsf{NP}}

\newcommand{\Ftwo}{\mathbb F_2}
\newcommand{\Bset}{\mathbb B}
\newcommand{\Sset}{\mathbb S}

\newcommand{\enc}{\operatorname{enc}}

\newcommand{\dbn}[2]{%
  [\![#1,#2]\!]%
}
\newcolumntype{L}[1]{>{\raggedright\arraybackslash}p{#1}}

\begin{document}

% - - - - - - - - - - - - - - - - - - - - - - - - - - - - - - - - - - - - - - - - - - - - -
%.     Frontmatter
% - - - - - - - - - - - - - - - - - - - - - - - - - - - - - - - - - - - - - - - - - - - - -

\title{On the Support Geometry and Complexity of LABS}
\titlerunning{On the Support Geometry and Complexity of LABS}

\author{William Cruz-Santos}
\authorrunning{William Cruz-Santos}

\institute{Tecnológico Nacional de México Campus Querétaro, Av Tecnológico S/N, Centro Histórico, 76000 Santiago de Querétaro, Qro. México \at \email{williamdelacruz@gmail.com, william.dd@queretaro.tecnm.mx}}
\date{}

\maketitle

\begin{abstract}
LABS are a fundamental combinatorial optimization problem arising in sequence design and statistical physics, yet the formal complexity of the standard fixed-geometry problem remains unresolved. We study how zero-energy complexity changes when support geometry is included as part of the input and then progressively restricted. For unary standard LABS, we establish membership in $\NP$ and a sparsity/Mahaney consequence: Karp NP-hardness would imply $\Ptime=\NP$. For Pair-Set-LABS and Masked-LABS, zero-energy feasibility is polynomial-time decidable when every pair group or masked row has at most three terms and is NP-complete when four terms are allowed. NP-completeness persists when active supports are restricted to contiguous intervals.

\keywords{Low autocorrelation binary sequences \and Computational complexity \and NP-completeness \and Combinatorial optimization}
\end{abstract}

% - - - - - - - - - - - - - - - - - - - - - - - - - - - - - - - - - - - - - - - - - - - - -
%.     Main text
% - - - - - - - - - - - - - - - - - - - - - - - - - - - - - - - - - - - - - - - - - - - - -

\section{Introduction}\label{sec:introduction}

The Low Autocorrelation Binary Sequences (LABS) problem seeks binary spin sequences $s\in\{-1,+1\}^N$ that minimize the sum of squared aperiodic autocorrelations. The objective, which appears as the energy of the Bernasconi model \cite{Bernasconi1987}, has long been studied in sequence design, combinatorial optimization, and statistical physics. Exact optimization is computationally demanding: Packebusch and Mertens give an exact algorithm with running time $\Theta(N1.73^N)$ and compute all optimal sequences through $N=66$ \cite{PackebuschMertens}, while larger lengths have been approached mainly by heuristic and restricted-space methods such as warning propagation and parallel GPU-based search \cite{Kotsireas2017,BoskovicHerzogBrest2024,ZhangEtAl2025,PsenicnikEtAl2026}. 

These computational results extend the range of LABS instances that can be studied, but they do not establish the complexity of a different question: whether there exists a binary sequence whose energy does not exceed a given bound. This is the decision problem considered here for standard LABS. An important feature of this problem is that the sequence length completely determines which products contribute to the autocorrelations. These contributions cannot be independently selected or removed. Furthermore, zero energy is impossible for sequences of length at least two, since one of the autocorrelations always consists of a single nonzero product. The standard decision problem therefore requires a general energy bound. We also investigate variants in which the input can specify how pair products are grouped, which terms within a correlation are retained, or which complete correlations are included. In these variants, zero energy may become attainable, and the selected terms determine the constraints that a sequence must satisfy. These observations motivate the central question of the paper: how does the complexity of zero-energy feasibility change as input-dependent support patterns are subjected to progressively stronger restrictions on their cardinality and geometry?

\subsection{Related work}\label{sec:related-work}
Prior work relevant to the present study includes models that modify interaction structure or interaction range. Bouchaud and M\'ezard analyze the deterministic Bernasconi model through an auxiliary random-connectivity system to study its glassy behavior \cite{BouchaudMezard}. Liers et al. introduce a non-disordered model that minimizes spin correlations within local blocks of a tunable interaction range $R$ \cite{Liers2009}. For the one-dimensional model, they provide an exact transfer-matrix algorithm with running time $O(N2^{2(R-3)})$, which is linear in $N$ for fixed $R$. These approaches differ from Masked-LABS, where the input explicitly specifies which pair interactions contribute to each lag correlation. Our analysis instead addresses the worst-case complexity of zero-energy feasibility under restrictions on support cardinality and geometry.

Another line of research concerns the selection and evaluation of autocorrelation terms. Song, Babu, and Palomar study weighted integrated sidelobe-level objectives, assigning weights to complete aperiodic autocorrelations at individual lags \cite{SongBabuPalomar}. Choosing binary weights provides a direct comparison with complete-lag selection in terms of which lag correlations are retained, although their work addresses sequence optimization rather than the corresponding zero-energy decision problem. Klapper and Goresky analyze autocorrelation sums of periodic binary sequences over finite windows at specified shifts \cite{KlapperGoresky}. These windowed correlations are structurally related to the contiguous supports of single-window Masked-LABS, but differ in their periodic setting and mathematical objectives.

Other relevant research addresses the representation and computational treatment of LABS. Walter studies binary linearization complexity and evaluates corresponding integer linear programming formulations for LABS \cite{Walter2024}. Computational approaches include QAOA \cite{ShaydulinEtAl2024}, quantum imaginary-time evolution applied to a PUBO formulation \cite{BauerAlamOstrowskiSiopsis2024}, and comparisons between direct PUBO and QUBO formulations \cite{IkeuchiMatsudaTanaka2026}. General Ising/QUBO modeling and quadratization methods provide the relevant representation background \cite{Lucas2014,GloverKochenbergerDu2019,AnthonyBorosCramaGruber2016,DragoiBaiardiEgger2026}. These studies address representation and computational methods rather than Karp decision complexity.

\subsection{Contributions and organization}
In this paper, we first formalize unary standard LABS and derive the sparsity/Mahaney barrier for Karp NP-hardness. We then study input-dependent support under progressively stronger geometric restrictions.

For the Pair-Set-LABS and Masked-LABS variants, zero-energy feasibility is polynomial-time decidable when every pair group or masked row has at most three terms and NP-complete when four terms are allowed. For the Masked-LABS variant, the NP-hardness construction respects the lag-diagonal restriction. Requiring each active support to be one contiguous interval preserves the corresponding classification: windows of length at most three are polynomial-time decidable, while allowing length four gives NP-completeness.

When every active lag retains its full LABS correlation, odd sequence lengths admit an exact parity classification. For even sequence length, any active odd lag makes zero energy impossible, while the remaining case with only even active lags is exactly equivalent to the selected-lag complementary-pair feasibility problem ($\SLCP$), defined in Section~\ref{sec:complete-lag}. Its general complexity remains open. Table~\ref{tab:complexity-landscape} summarizes the main complexity results across the support geometries considered in this paper.

\begin{table}[t]
\caption{Complexity landscape across the support geometries considered in this paper.}
\label{tab:complexity-landscape}
\centering
\small
\setlength{\tabcolsep}{3pt}
\begin{tabular}{@{}L{0.18\textwidth}L{0.25\textwidth}L{0.22\textwidth}L{0.26\textwidth}@{}}
\hline
Formulation & Geometry & Result & Limit \\
\hline

Standard LABS
& Support fixed by $N$
& In $\NP$
& Karp NP-hardness would imply $\Ptime=\NP$. \\

Pair-Set-LABS
& Arbitrary pair groups
& At most $3$ terms: $\Ptime$
& Allowing $4$ terms: $\NP$-complete. \\

Masked-LABS
& One lag diagonal per row
& At most $3$ terms: $\Ptime$
& Allowing $4$ terms: $\NP$-complete. \\

Single-window
& One contiguous interval per active lag
& Length at most $3$: $\Ptime$
& Allowing length $4$: $\NP$-complete. \\

Complete-lag
& Full correlation at each active lag
& Odd $N$: exact parity classification
& Even $N$: odd active lags are infeasible; otherwise equivalent to $\SLCP$. \\
\hline
\end{tabular}
\end{table}

The organization of this paper is as follows. Section~\ref{sec:preliminaries} fixes notation, the objective, and the encoding and complexity conventions. Section~\ref{sec:fixed-baseline} treats standard LABS. Section~\ref{sec:local-support} studies pair groups and masked lag-diagonal supports, and Section~\ref{sec:contiguous} imposes the additional requirement of contiguous support. Section~\ref{sec:complete-lag} analyzes the case in which every active lag retains its full correlation. Section~\ref{sec:discussion} summarizes the conclusions.

\section{Preliminaries}\label{sec:preliminaries}

Let $\Bset:=\{0,1\}$ be the set of Boolean values, and let $\Sset:=\{-1,+1\}$ be the set of spins. We write
\begin{equation}
\Ftwo:=\mathbb Z/2\mathbb Z
\end{equation}
for the field with two elements; its addition and multiplication are taken modulo $2$. Unless stated otherwise, sequence and vertex indices are positive integers.

For integers $i,j\in\mathbb Z$, define
\begin{equation}
\dbn{i}{j}:=\{r\in\mathbb Z:i\le r\le j\}.
\end{equation}
Thus $\dbn{i}{j}=\varnothing$ when $i>j$. Whenever a maximum of support cardinalities is taken over an empty index set, we use the convention $\max\varnothing:=0$.

This interval notation is reserved for integer indices and vertex labels. The positions of a sequence of length $N$ are indexed by
\begin{equation}
V_N:=\dbn{1}{N}.
\end{equation}

An edge on a finite vertex set $V$ is written as an unordered pair $\{u,v\}$, where $u,v\in V$ and $u\ne v$. Whenever a coefficient is indexed by the endpoints of an edge $\{i,j\}$, we order the indices so that $i<j$.

For a sequence
\begin{equation}
s=(s_1,\ldots,s_N)\in\Sset^N,
\end{equation}
define the aperiodic autocorrelation at lag $k\in\dbn{1}{N-1}$ by
\begin{equation}\label{eq:std-corr}
C_k(s)=\sum_{i=1}^{N-k}s_is_{i+k},
\end{equation}
and define the LABS energy by
\begin{equation}\label{eq:std-energy}
E_N(s)=\sum_{k=1}^{N-1}C_k(s)^2.
\end{equation}

The optimization problem is to compute a sequence minimizing $E_N(s)$.

For a fixed lag $k$, define the support
\begin{equation}
E_k
:=\bigl\{\{i,i+k\}:i\in\dbn{1}{N-k}\bigr\}.
\end{equation}

Eq.~(\ref{eq:std-corr}) can then be written as
\begin{equation}
C_k(s)=\sum_{\{i,j\}\in E_k}s_is_j.
\end{equation}

Figure~\ref{fig:support-geometries} gives a pair-index-matrix view of standard LABS and of the support geometries considered later in the paper. We call $E_k$ the \emph{lag-$k$ diagonal}, or simply the $k$th diagonal; for fixed $N$, each such support is completely determined.

\begin{figure}[t]
\centering
\includegraphics[width=0.99\linewidth]{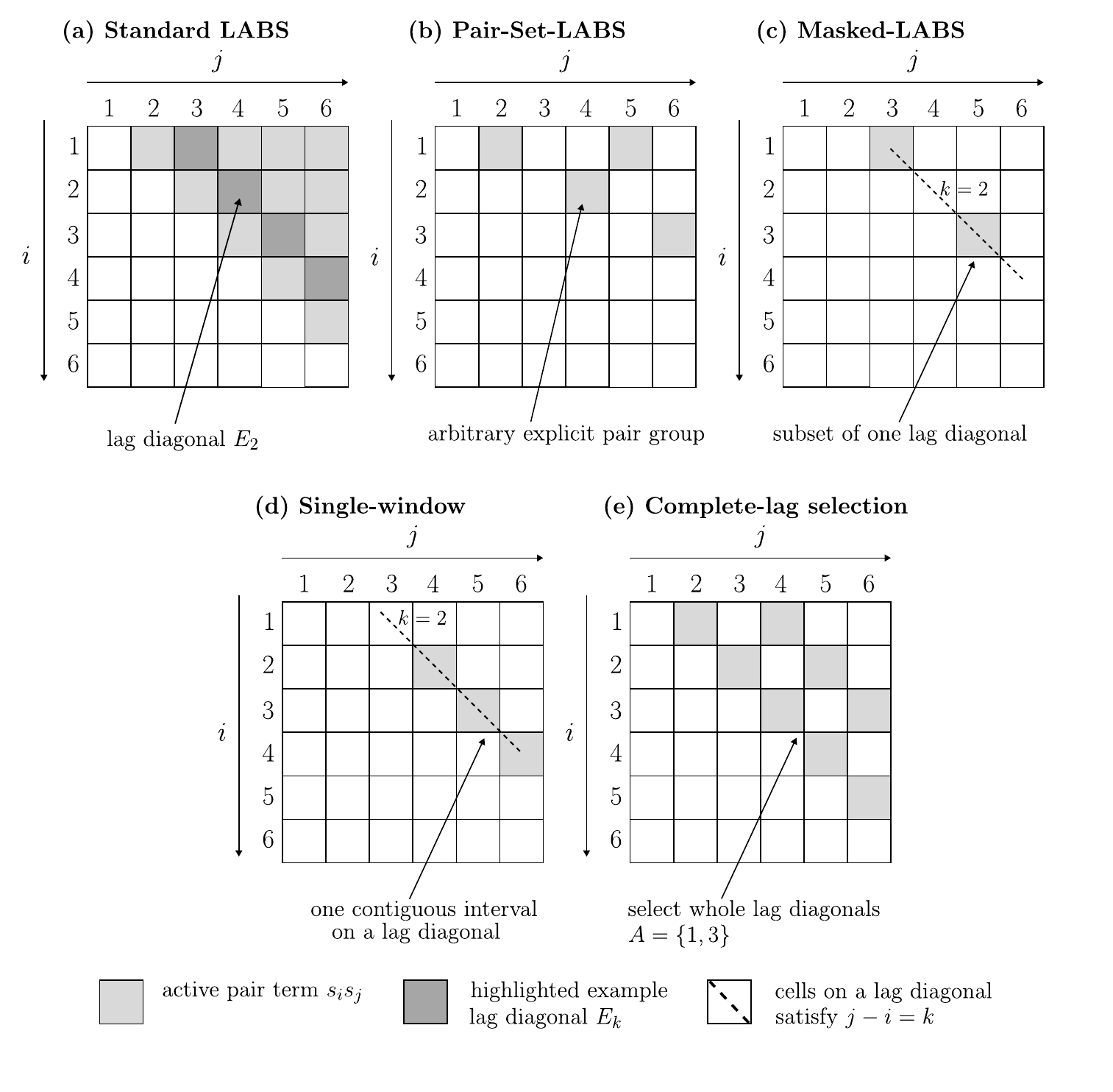}
\caption{Pair-index-matrix view of the support geometries considered in this paper. Each cell $(i,j)$ with $i<j$ represents the pair term $s_i s_j$, and cells on a common superdiagonal satisfy $j-i=k$ and therefore correspond to a fixed lag $k$. (a) Standard LABS contains all such pair terms; one complete lag diagonal is highlighted. (b) A Pair-Set-LABS group may contain arbitrary explicit pairs. (c) A Masked-LABS row may select an arbitrary subset of one lag diagonal. (d) Single-window support further requires the selected subset on each active lag to be one contiguous interval. (e) Complete-lag selection activates or deactivates whole lag diagonals. The $6\times6$ matrices are illustrative.}
\label{fig:support-geometries}
\end{figure}

\subsection{Encoding and complexity}\label{sec:encodings}
Let $\Sigma$ be a fixed finite alphabet containing the symbols $0$ and $1$, and let $\Sigma^*$ denote the set of all finite strings over $\Sigma$.

A decision problem $X$ consists of a set $\mathcal I_X$ of admissible instances and a YES/NO predicate on $\mathcal I_X$. For $I\in\mathcal I_X$, let $\enc_X(I)\in\Sigma^*$ denote a fixed canonical string representation of $I$ that can be parsed in polynomial time. A string is \emph{malformed} if it is not the canonical representation of any instance in $\mathcal I_X$; malformed strings are rejected.

A positive integer $n$ is represented in unary by the string $1^n$. A nonnegative integer $b$ is represented in binary by its shortest base-$2$ representation, with $0$ represented by the single symbol $0$.

The language associated with $X$ is
\begin{equation}
L_X:=\{\enc_X(I):I\in\mathcal I_X\text{ and }I\text{ is a YES-instance}\}.
\end{equation}

A \emph{witness}, or \emph{certificate}, for a YES-instance is auxiliary data that certifies membership in the associated language. A language $L$ belongs to $\Ptime$ if membership in $L$ can be decided in polynomial time. It belongs to $\NP$ if there is a polynomial-time verifier such that every YES-instance $x\in L$ has a certificate of length polynomial in $|x|$ that the verifier accepts, while no certificate is accepted for a NO-instance. Here $|x|$ denotes the length of the string $x$.

A polynomial-time many-one, or Karp, reduction from a language $L$ to a language $L'$ is a polynomial-time computable map $f$ such that
\begin{equation}
x\in L\quad\text{iff}\quad f(x)\in L'.
\end{equation}

A language $L$ is \emph{NP-hard} if $A\le_m^p L$ for every $A\in\NP$. It is \emph{NP-complete} if, in addition, $L\in\NP$. All NP-hardness and NP-completeness statements in this paper use Karp reductions.

An optimization problem $X$ consists of a set $\mathcal I_X$ of admissible instances, a set $\mathcal F_X(I)$ of feasible solutions for each $I\in\mathcal I_X$, and an objective function $E_I:\mathcal F_X(I)\to\mathbb R$ to be minimized.

The associated \emph{threshold decision problem} asks, for an instance $I$ and threshold $B$, whether there exists $z\in\mathcal F_X(I)$ such that $E_I(z)\le B$. The case $B=0$ is called \emph{zero-energy feasibility}.

\section{Standard LABS}\label{sec:fixed-baseline}

An instance of the standard LABS decision problem consists of a sequence length $N$, encoded in unary, and a nonnegative energy threshold $B$, encoded in binary. We denote its canonical representation by $\enc_{\LABS}(N,B)$. The corresponding decision language is
\begin{equation}\label{eq:labs-dec-unary}
L_{\mathrm{LABS}}^{\mathrm{unary}}:=\{\enc_{\LABS}(N,B):N\ge2,\ B\in\mathbb Z_{\ge0},\ \exists s\in\Sset^N,\ E_N(s)\le B\}.
\end{equation}

The exact-search literature does not determine the Karp complexity of this language. Although Kotsireas et al. describe LABS as NP-hard in their introduction \cite{Kotsireas2017}, they do not state an explicit polynomial-time many-one reduction for the standard one-parameter family.

\begin{proposition}\label{prop:labs-in-np}
The language $L_{\mathrm{LABS}}^{\mathrm{unary}}$ belongs to $\NP$.
\end{proposition}
\begin{proof}
A certificate is the spin vector $s\in\Sset^N$. Given $s$, compute every autocorrelation $C_k(s)$ by evaluating $N-k$ products and additions. There are $\binom{N}{2}=O(N^2)$ pair products in total.

Each $C_k(s)$ has magnitude at most $N-k$, so the arithmetic uses $O(\log N)$-bit integers. Squaring and summing the $N-1$ correlations therefore takes polynomial time in $N$ and in the bit length of $B$. The verifier then compares the computed energy with the threshold $B$, accepting exactly when $E_N(s)\le B$ and rejecting otherwise.
\qed
\end{proof}

\subsection{Normalization and sparsity}

To examine the implications of unary encoding for NP-hardness, we first establish a sparse normalized form of the LABS language. This allows us to apply Mahaney's theorem \cite{Mahaney}. A language $L\subseteq\Sigma^*$ is called \emph{sparse} if there exists a polynomial $p$ such that, for every $\ell\ge0$,
\begin{equation}
\left|\{x\in L:|x|\le\ell\}\right|\le p(\ell).
\end{equation}

Thus sparsity is a statement about the density of YES-strings as a function of encoding length; it does not by itself provide an algorithm for deciding membership. The following lemmas establish the normalization and sparsity properties of the unary LABS language.

\begin{lemma}[Energy bound]\label{labs-lem:energybound}
For every $N\ge1$ and every $s\in\Sset^N$,
\begin{equation}\label{labs-eq:UN}
0\le E_N(s)\le U(N):=\frac{N(N-1)(2N-1)}6.
\end{equation}
\end{lemma}
\begin{proof}
For lag $k$, the autocorrelation $C_k(s)$ is a sum of $N-k$ signs, hence
$|C_k(s)|\le N-k$. Therefore
\begin{equation}
E_N(s)
=\sum_{k=1}^{N-1}C_k(s)^2
\le\sum_{k=1}^{N-1}(N-k)^2
=\frac{N(N-1)(2N-1)}6
=U(N).
\end{equation}
\qed
\end{proof}

Since $E_N(s)\le U(N)$ for every $s\in\Sset^N$, any well-formed instance with $B\ge U(N)$ is a YES-instance. For the normalization below, thresholds with $B>U(N)$ are therefore redundant, whereas $B=U(N)$ remains admissible.

We define
\begin{equation}
L_{\mathrm{LABS}}^{\mathrm{cap}}:=\{\enc_{\LABS}(N,B):0\le B\le U(N),\ \exists s\in\Sset^N,\ E_N(s)\le B\}.
\end{equation}

The superscript ``cap'' refers to the restriction $B\le U(N)$. Malformed encodings and well-formed instances with $B>U(N)$ are not members of $L_{\mathrm{LABS}}^{\mathrm{cap}}$.

\begin{lemma}[Threshold normalization]\label{labs-lem:normalize}
There is a polynomial-time many-one reduction
\begin{equation}
L_{\mathrm{LABS}}^{\mathrm{unary}}\le_m^p L_{\mathrm{LABS}}^{\mathrm{cap}}.
\end{equation}
\end{lemma}
\begin{proof}
On a well-formed input $x=\enc_{\LABS}(N,B)$, compute $U(N)$ using Eq.~(\ref{labs-eq:UN}). If $B\le U(N)$, output $x$ unchanged. If $B>U(N)$, output the fixed YES-instance $y_{\rm yes}:=\enc_{\LABS}(2,1)$. Since every length-two sequence has energy one, $y_{\rm yes}\in L_{\mathrm{LABS}}^{\mathrm{cap}}$. By Lemma~\ref{labs-lem:energybound}, every well-formed input with $B>U(N)$ belongs to $L_{\mathrm{LABS}}^{\mathrm{unary}}$, so the map preserves membership. Malformed inputs may be sent to the fixed NO-instance $\enc_{\LABS}(2,0)$. The normalization is polynomial-time because $N$ is explicitly unary and $U(N)$ has $O(\log N)$ bits.
\qed
\end{proof}

\begin{lemma}[Sparsity]\label{labs-lem:sparse}
The normalized language $L_{\mathrm{LABS}}^{\mathrm{cap}}$ is sparse. More precisely, the number of strings in $L_{\mathrm{LABS}}^{\mathrm{cap}}$ of length at most $\ell$ is $O(\ell^4)$ under any fixed canonical self-delimiting encoding.
\end{lemma}
\begin{proof}
Any normalized string of total length at most $\ell$ has $N\le\ell$ because it contains $1^N$. For a fixed $N$, the allowed thresholds satisfy $0\le B\le U(N)=O(N^3)$, so there are at most $U(N)+1=O(N^3)$ corresponding encoded instances. Therefore
\begin{equation}
\left|\{x\in L_{\mathrm{LABS}}^{\mathrm{cap}}:|x|\le\ell\}\right|
\le
\sum_{N=1}^{\ell}(U(N)+1)
=O(\ell^4).
\end{equation}
\qed
\end{proof}

\subsection{Karp-hardness barrier}\label{subsec:karp-barrier}
Mahaney's theorem states that if a sparse language is NP-hard under polynomial-time many-one reductions, then $\Ptime=\NP$ \cite{Mahaney}. Equivalently, assuming $\Ptime\ne\NP$, no sparse language in $\NP$ can be NP-hard under Karp reductions. Modern expositions often state this form directly \cite{Grochow}.

\begin{theorem}[Unary LABS hardness barrier]\label{thm:mahaney-labs}
If $L_{\mathrm{LABS}}^{\mathrm{unary}}$ is NP-hard under polynomial-time many-one reductions, then $\Ptime=\NP$.

Equivalently, if $\Ptime\ne\NP$, then $L_{\mathrm{LABS}}^{\mathrm{unary}}$ is not NP-hard under polynomial-time many-one reductions.
\end{theorem}
\begin{proof}
Assume $L_{\mathrm{LABS}}^{\mathrm{unary}}$ is NP-hard under polynomial-time many-one reductions. By Lemma~\ref{labs-lem:normalize}, every language in $\NP$ reduces to $L_{\mathrm{LABS}}^{\mathrm{cap}}$ by composing its reduction to $L_{\mathrm{LABS}}^{\mathrm{unary}}$ with the reduction of Lemma~\ref{labs-lem:normalize}. Hence $L_{\mathrm{LABS}}^{\mathrm{cap}}$ is NP-hard.

Moreover, $L_{\mathrm{LABS}}^{\mathrm{cap}}\in\NP$: the verifier of Proposition~\ref{prop:labs-in-np} applies unchanged, with the additional condition $B\le U(N)$ checked in polynomial time. Therefore $L_{\mathrm{LABS}}^{\mathrm{cap}}$ is NP-complete.

By Lemma~\ref{labs-lem:sparse}, $L_{\mathrm{LABS}}^{\mathrm{cap}}$ is sparse, so Mahaney's theorem yields $\Ptime=\NP$ under the assumed NP-hardness.
\qed
\end{proof}

Theorem~\ref{thm:mahaney-labs} applies to standard LABS with unary encoding and polynomial-time many-one reductions; it establishes neither a polynomial-time algorithm nor a corresponding hardness barrier for other encodings or reduction notions. Furthermore, NP-completeness of the input-dependent support variants does not imply NP-hardness of standard LABS without an explicit reduction.

\section{Support cardinality}\label{sec:local-support}

%We now study how local modifications of the fixed support geometry of standard LABS affect zero-energy complexity.

\subsection{Pair-Set-LABS}
Pair-Set-LABS allows arbitrary explicit groups of pair interactions. This unrestricted pair-support geometry is illustrated in Fig.~\ref{fig:support-geometries}(b). Let $V_n:=\dbn{1}{n}$ be the vertex set. A structural instance is an ordered family
\begin{equation}
\mathcal E=(E_1,\ldots,E_m),
\qquad
E_a\subseteq\bigl\{\{u,v\}:u,v\in V_n,\ u\ne v\bigr\}.
\end{equation}

Its canonical threshold-instance encoding is
\begin{equation}\label{eq:pair-encoding}
\enc_{\mathrm{PairSet}}(n,\mathcal E,B),
\end{equation}
where $n$ is represented in unary, every $E_a$ is stored as a canonical explicit edge list, and $B$ is represented in binary. Writing $L(\mathcal E):=\sum_a|E_a|$, its bit length is $O\!\left(n+m+L(\mathcal E)\log(n+1)+\log(B+2)\right)$.

For $s\in\Sset^n$ and a group index $a\in\dbn{1}{m}$, define
\begin{equation}
C_a^{\mathcal E}(s)
:=
\sum_{\{i,j\}\in E_a}s_is_j.
\end{equation}

The total energy is
\begin{equation}
E_{\mathcal E}(s)
:=
\sum_{a=1}^{m}\left(C_a^{\mathcal E}(s)\right)^2.
\end{equation}

The optimization problem minimizes $E_{\mathcal E}(s)$ over $s\in\Sset^n$. Under the encoding in Eq.~(\ref{eq:pair-encoding}), the corresponding threshold decision language is
\begin{equation}
L_{\mathrm{PairSet}}:=\{\enc_{\mathrm{PairSet}}(n,\mathcal E,B):\exists s\in\Sset^n,\ E_{\mathcal E}(s)\le B\}.
\end{equation}

Define the maximum group cardinality
\begin{equation}
\rho(\mathcal E):=\max_{a\in\dbn{1}{m}}|E_a|.
\end{equation}

Because the energy is a sum of squares, $E_{\mathcal E}(s)=0$ if and only if $C_a^{\mathcal E}(s)=0$ for every $a\in\dbn{1}{m}$.

\begin{proposition}\label{prop:pair-np}
The language $L_{\mathrm{PairSet}}$ belongs to $\NP$.
\end{proposition}
\begin{proof}
A certificate is $s\in\Sset^n$.  Its length is $n$, hence polynomial because the input contains $1^n$.  For each edge $\{i,j\}\in E_a$, compute the product $s_is_j$ and add it to the corresponding group correlation.  This requires $O(L(\mathcal E))$ products and additions.  Each $|C_a^{\mathcal E}(s)|\le |E_a|$, so the integer arithmetic has polynomial bit length.  Squaring the correlations, summing them, and comparing the result with the binary threshold $B$ are therefore polynomial-time operations in the bit length of the input.
\qed
\end{proof}

\subsection{Pair-Set cardinality}

To analyze zero-energy feasibility, we first examine how group cardinality affects the correlation constraints. Every group correlation must vanish; odd-cardinality groups cannot satisfy this condition by parity, while two-pair groups yield affine constraints over $\Ftwo$.

\begin{lemma}[Odd-cardinality]\label{lem:pair-odd}
If $|E_a|$ is odd, then $C_a^{\mathcal E}(s)\ne0$ for every assignment $s\in\Sset^n$.
\end{lemma}
\begin{proof}
Every product $s_is_j$ belongs to $\Sset$. A sum of an odd number of odd integers is odd, so $C_a^{\mathcal E}(s)$ is an odd integer and cannot equal zero.
\qed
\end{proof}

Consequently, under the bound $\rho(\mathcal E)\le3$, groups of cardinality one and three are excluded from every zero-energy YES-instance. Only cardinality-two groups require further analysis.

\begin{lemma}[Two-pair parity]\label{lem:pair-two}
Let $E_a=\{\{i,j\},\{k,\ell\}\}$. Under the Boolean encoding $s_v=(-1)^{x_v}$ with $x_v\in\Ftwo$, the equation $C_a^{\mathcal E}(s)=0$ is equivalent to $x_i+x_j+x_k+x_\ell\equiv1\pmod2$, with any repeated vertex cancelling modulo two.
\end{lemma}
\begin{proof}
The zero equation $s_is_j+s_ks_\ell=0$ is equivalent to $s_is_js_ks_\ell=-1$. Substituting $s_v=(-1)^{x_v}$ gives $(-1)^{x_i+x_j+x_k+x_\ell}=-1$, equivalently $x_i+x_j+x_k+x_\ell\equiv1\pmod2$. If a vertex occurs twice among the two pairs, its exponent occurs twice and therefore cancels modulo two in $\Ftwo$.
\qed
\end{proof}

\begin{proposition}[Tractability at cardinality three]\label{prop:pair-three}
The zero-energy decision version of Pair-Set-LABS is polynomial-time decidable on instances satisfying $\rho(\mathcal E)\le3$.
\end{proposition}
\begin{proof}
If any nonempty group has cardinality one or three, Lemma~\ref{lem:pair-odd} gives an immediate NO answer. After removing empty groups, every remaining group must therefore have cardinality two. Lemma~\ref{lem:pair-two} converts every remaining zero equation into one affine linear equation over $\Ftwo$. The instance is feasible exactly when the resulting affine system is consistent; Gaussian elimination on its augmented matrix over $\Ftwo$ decides this in polynomial time.
\qed
\end{proof}

At cardinality four, a group can encode an exact-one constraint on three logical variables; a reduction from Distinct Positive 1-in-3-SAT from Appendix~\ref{app:distinct-p13} then yields NP-completeness.

\begin{theorem}[Cardinality-four NP-completeness]\label{thm:pair-threshold}
Under the encoding in Eq.~(\ref{eq:pair-encoding}), zero-energy feasibility for Pair-Set-LABS restricted to $\rho(\mathcal E)\le4$ is NP-complete.
\end{theorem}
\begin{proof}
Membership in $\NP$ follows from Proposition~\ref{prop:pair-np}, so it remains to prove NP-hardness. We reduce from Distinct Positive 1-in-3-SAT, which is NP-complete by Lemma~\ref{pair-src-lem:distinct-p13-npc}.

\emph{Clause construction.} Let the source variables be $x_1,\ldots,x_t$, with clauses $R(x_i,x_j,x_\ell)$ in which $i,j,\ell$ are pairwise distinct. Construct one Pair-Set vertex for every source variable and two additional vertices $A,B$, distinct from all source-variable vertices. For each clause $R(x_i,x_j,x_\ell)$, create the four-edge set
\begin{equation}
E_a=\{\{A,B\},\{A,i\},\{A,j\},\{A,\ell\}\},
\end{equation}
whose edges are distinct because $i,j,\ell$ are pairwise distinct and $A,B$ are new vertices.

The correlation term for $E_a$ can be written as
\begin{equation}
\begin{aligned}
C_a^{\mathcal E}(s)
&=s_As_B+s_As_i+s_As_j+s_As_\ell\\
&=s_As_B\bigl(1+\tau_i+\tau_j+\tau_\ell\bigr)
\end{aligned}
\label{pair-thm-eq:tau-factor}
\end{equation}
where $\tau_v:=s_Bs_v\in\Sset$.

Interpreting $\tau_v=+1$ as Boolean true and $\tau_v=-1$ as false, suppose exactly $r\in\{0,1,2,3\}$ of $\tau_i,\tau_j,\tau_\ell$ equal $+1$. Then $\tau_i+\tau_j+\tau_\ell=2r-3$, so the bracket in Eq.~(\ref{pair-thm-eq:tau-factor}) equals $2(r-1)$. Therefore
\begin{equation}\label{pair-thm-eq:zero-iff-one}
C_a^{\mathcal E}(s)=0
\quad\text{iff}\quad
r=1.
\end{equation}

Thus the correlation associated with $E_a$ vanishes exactly when the corresponding source clause has one true variable under the induced Boolean interpretation.

\emph{Completeness.} Given a satisfying source assignment, set $s_B=+1$, set $s_v=+1$ for true source variables and $s_v=-1$ for false source variables, and choose either sign for $s_A$. Every source clause has exactly one true variable, so Eq.~(\ref{pair-thm-eq:zero-iff-one}) gives zero correlation for every target group and hence zero total energy.

\emph{Soundness.} Given a zero-energy target assignment, every group correlation vanishes. Define the source truth assignment by declaring $x_v$ true exactly when $\tau_v=s_Bs_v=+1$. Eq.~(\ref{pair-thm-eq:zero-iff-one}) then implies that every source clause contains exactly one true variable.

\emph{Polynomial size.} If the source has $t$ variables and $q$ clauses, the target has $n=t+2$ vertices and $q$ groups, each containing four edges. Since $n=t+2$ is encoded in unary, its representation has length $t+2$, and every edge endpoint requires $O(\log(t+2))$ bits. Thus the target encoding has polynomial length and the construction is computable in polynomial time. This proves NP-hardness, and membership in $\NP$ gives NP-completeness.
\qed
\end{proof}

\subsection{Masked-LABS}\label{ch:masked}
\emph{Masked-LABS} allows the input to select, for each lag $k$, a subset of the pair products $s_is_{i+k}$, $i\in\dbn{1}{N-k}$, that contribute to the autocorrelation $C_k(s)$. The selected terms retain the defining lag-$k$ relation between their endpoints, while the complete lag diagonal need not be present. This geometry is illustrated in Fig.~\ref{fig:support-geometries}(c).

For a length-$N$ instance, define
\begin{equation}
\Omega_N:=\{(k,i):k\in\dbn{1}{N-1},\ i\in\dbn{1}{N-k}\}.
\end{equation}

A {\em mask} is the explicit sparse support set $M\subseteq\Omega_N$, represented as a canonical lexicographically sorted list with no duplicates. For each lag $k$, define
\begin{equation}
I_k^M:=\{i\in\dbn{1}{N-k}:(k,i)\in M\},
\qquad
L(M):=|M|=\sum_{k=1}^{N-1}|I_k^M|.
\end{equation}

The canonical threshold-instance encoding is
\begin{equation}\label{eq:masked-encoding}
\enc_{\mathrm{Masked}}(N,M,B),
\end{equation}
where $N$ is represented in unary, $M$ by its sparse canonical support list, and $B$ in binary. Its bit length is $O\!\left(N+L(M)\log(N+1)+\log(B+2)\right)$.

The masked autocorrelation at lag $k$ is
\begin{equation}
C_k^M(s)
:=\sum_{i\in I_k^M}s_is_{i+k},
\end{equation}
and the Masked-LABS energy is
\begin{equation}
E_M(s)
:=\sum_{k=1}^{N-1}\bigl(C_k^M(s)\bigr)^2.
\end{equation}

Thus the superscript $M$ on $C_k^M$ means ``correlation under mask $M$,'' and the subscript $M$ on $E_M$ records that the objective is determined by the entire mask.

The Masked-LABS optimization problem minimizes $E_M(s)$ over $s\in\Sset^N$.

\begin{definition}[Masked-LABS decision]\label{def:masked-dec}
Under the encoding in Eq.~(\ref{eq:masked-encoding}), define the language
\begin{equation}
L_{\mathrm{Masked}}:=\{\enc_{\mathrm{Masked}}(N,M,B):\exists s\in\Sset^N,\ E_M(s)\le B\}.
\end{equation}

\end{definition}

Since $E_M(s)$ is a sum of squares, $E_M(s)=0$ if and only if $C_k^M(s)=0$ for every $k$ with $I_k^M\ne\varnothing$.

Define the \emph{masked-row cardinality} at lag $k$ as $|I_k^M|$, the number of retained pair terms in that row, and define
\begin{equation}
\rho(M):=\max_{k\in\dbn{1}{N-1}}|I_k^M|.
\end{equation}

Each masked row is itself a set of pair products. Therefore the parity arguments from Pair-Set-LABS apply row by row, with the additional restriction that all pairs in row $k$ have the same difference $k$.

\begin{proposition}[Tractability at cardinality three]\label{prop:masked-three}
The zero-energy decision version of Masked-LABS is in $\Ptime$ on instances satisfying $\rho(M)\le3$.
\end{proposition}
\begin{proof}
If some nonempty masked row has cardinality one or three, its correlation is a sum of an odd number of $\pm1$ terms and cannot be zero. Otherwise every nonempty masked row has cardinality two. Writing $s_v=(-1)^{x_v}$, the zero condition for a cardinality-two masked row, $s_is_{i+k}+s_js_{j+k}=0$, is equivalent to the affine constraint $x_i+x_{i+k}+x_j+x_{j+k}=1\pmod2$.

Collecting these constraints yields an affine linear system $Ax=b$ over $\Ftwo$, with one variable $x_v$ for each spin position and one constraint for each nonempty masked row. Applying Gaussian elimination over $\Ftwo$ to the augmented matrix $[A\mid b]$ determines whether the system is consistent: a row $0\;0\;\cdots\;0\mid1$ gives a contradiction, while
otherwise the system has a solution and hence a zero-energy spin assignment. Since the system has polynomially many variables and constraints in the input size, the procedure runs in polynomial time.
\qed
\end{proof}

\subsection{Cardinality-four hardness}\label{sec:masked-hardness}

Here we establish NP-hardness for Masked-LABS using cardinality-two equality rows and cardinality-four clause rows consisting of disjoint parallel edges.

\begin{lemma}[NOT row]\label{lem:spin-not-wire}
For spin positions $a,a+k,a+2k$, the zero equation $s_as_{a+k}+s_{a+k}s_{a+2k}=0$ is equivalent to $s_{a+2k}=-s_a$.
\end{lemma}
\begin{proof}
Factoring the left-hand side gives $s_{a+k}(s_a+s_{a+2k})=0$. Since $s_{a+k}\in\Sset$, this holds if and only if $s_a+s_{a+2k}=0$, equivalently $s_{a+2k}=-s_a$.
\qed
\end{proof}

If two NOT rows share the same newly introduced endpoint $z$, their zero equations enforce $s_z=-s_p$ and $s_z=-s_q$, and therefore $s_p=s_q$. Lemma~\ref{lem:unique-lag-embed} formalizes this construction while assigning distinct lags, avoiding collisions among newly introduced spin positions, and keeping the sequence length polynomially bounded.

\begin{lemma}[Unique-lag equality embedding]\label{lem:unique-lag-embed}
Let $P_0\subset\mathbb Z_{>0}$ be a nonempty finite set of spin positions, all having the same parity, and let $F_0\subset\mathbb Z_{>0}$ be a finite set of lags that the new rows are not allowed to use. Given $R$ pairs $(p_r,q_r)\in P_0^2$ specifying constraints $s_{p_r}=s_{q_r}$, trivial constraints with $p_r=q_r$ may be discarded. Every remaining constraint can be implemented by two cardinality-two masked rows of the form in Lemma~\ref{lem:spin-not-wire} such that:
\begin{enumerate}
\item every new row uses a lag outside $F_0$, and all new lags are distinct;
\item all newly introduced spin positions are distinct from previously used spin positions, except that the two rows implementing the same constraint share the endpoint $z_r$;
\item $z_r$ has the same parity as $p_r$ and $q_r$, so the required midpoint positions are integral; and
\item the largest spin position introduced is bounded by a polynomial in $\max P_0+|P_0|+|F_0|+R$.
\end{enumerate}
\end{lemma}
The proof is given in Appendix~\ref{app:masked-embedding}.

We now specify the positions used in the reduction from 2-2 Set Splitting \cite{Guruswami,SpielmanZhang2022}. Let the source instance have universe $U=\{v_1,\ldots,v_t\}$ and four-element sets $S_1,\ldots,S_m$, each containing four distinct elements. Since an instance with no sets is trivially satisfiable, we assume $m\ge1$ for the construction. Assign the global reference spin position $g:=2$ and the representative spin position $x_{v_j}:=2(j+1)$ to each $v_j\in U$, so all representative positions are even. For each $a\in\dbn{1}{m}$, set the even clause lag $K_a:=2(a+4)$; these lags are pairwise distinct and satisfy $K_a\ge10$. Let
\begin{equation*}
K_{\max}:=2(m+4), \quad W:=2(K_{\max}+8), \quad B_a:=2(t+2)+aW, \quad a\in\dbn{1}{m}.
\end{equation*}

Writing $S_a=\{v_{a,1},v_{a,2},v_{a,3},v_{a,4}\}$, define for $r\in\dbn{1}{4}$
\begin{equation*}
g_{a,r}:=B_a+2(r-1), \qquad x_{a,r}:=g_{a,r}+K_a.
\end{equation*}

\begin{lemma}[Clause placement]\label{lem:masked-clause-place}
The spin positions and lags defined above satisfy:
\begin{enumerate}
\item all designated coordinates are even;
\item the four edges $\{g_{a,r},x_{a,r}\}$, $r\in\dbn{1}{4}$, are pairwise vertex-disjoint and all have lag $K_a$;
\item different clause blocks are vertex-disjoint;
\item all clause lags are pairwise distinct; and
\item the largest designated coordinate is $O(t+m^2)$.
\end{enumerate}
\end{lemma}
The proof is given in Appendix~\ref{app:masked-embedding}.

\begin{theorem}[Cardinality-four NP-completeness]\label{thm:masked-four}
Zero-energy feasibility for Masked-LABS is NP-complete when every nonempty masked row has cardinality at most four.  Hardness holds even when every cardinality-four row consists of four disjoint parallel equal-lag edges and every cardinality-two row is a consecutive-edge NOT row.
\end{theorem}
\begin{proof}
Membership in $\NP$ follows from the unary-length support-list encoding: a certificate is a spin sequence $s\in\Sset^N$, and all selected pair products can be evaluated in polynomial time.

For NP-hardness, reduce from 2-2 Set Splitting \cite{Guruswami,SpielmanZhang2022}. Let the source instance have universe $U=\{v_1,\ldots,v_t\}$ and sets $S_1,\ldots,S_m$, each containing four distinct elements. A source assignment associates a sign in $\Sset$ with every element of $U$, and a set is satisfied exactly when its four signs split two-plus/two-minus.

\emph{Construction.} Use the reference spin position $g$, the representative spin positions $x_v$, the clause lags $K_a$, and the clause positions $g_{a,r},x_{a,r}$ defined above. By Lemma~\ref{lem:masked-clause-place}, all these positions have the required parity, the four edges $\{g_{a,r},x_{a,r}\}$, $r\in\dbn{1}{4}$, form a matching for each source set, the clause lags are pairwise distinct, and the largest position is $O(t+m^2)$.

For a source set $S_a=\{v_{a,1},v_{a,2},v_{a,3},v_{a,4}\}$, include the four edges $\{g_{a,r},x_{a,r}\}$, $r\in\dbn{1}{4}$, in the masked row at lag $K_a$. Impose the equality constraints $s_{g_{a,r}}=s_g$ and $s_{x_{a,r}}=s_{x_{v_{a,r}}}$ for every $a\in\dbn{1}{m}$ and $r\in\dbn{1}{4}$.

Let $P_0$ be the set consisting of $g$, all representative positions $x_v$, and all clause positions $g_{a,r},x_{a,r}$. Let $F_0:=\{K_1,\ldots,K_m\}$. The set $P_0$ contains $1+t+8m$ spin positions, all of the same parity, while $|F_0|=m$ and $\max P_0=O(t+m^2)$. Apply Lemma~\ref{lem:unique-lag-embed} to the $8m$ equality constraints. Each constraint is implemented by two cardinality-two masked rows whose lags are distinct from one another and from all clause lags. The lemma also guarantees that every newly introduced spin position and the final sequence length remain polynomially bounded in $t+m$.

After all equality constraints have been implemented, let $N$ be the largest spin position introduced and let $M\subseteq\Omega_N$ be the resulting mask. The target instance is $\enc_{\mathrm{Masked}}(N,M,0)$. Every nonempty masked row has cardinality two or four, so $\rho(M)\le4$.

If the equality constraints associated with source set $S_a$ hold, then its four-term correlation is
\begin{equation}\label{masked-eq:clause-reduced}
C_{K_a}^M(s)=\sum_{r=1}^4 s_{g_{a,r}}s_{x_{a,r}}=s_g\sum_{r=1}^4 s_{x_{v_{a,r}}}.
\end{equation}

Hence the clause row has zero correlation exactly when the four representative spins associated with $S_a$ split two-plus/two-minus.

\emph{Completeness.} Suppose the source instance has a valid 2-2 splitting. Assign each representative spin $s_{x_v}$ its source sign and choose either sign for the reference spin $s_g$. Set every clause position so that $s_{g_{a,r}}=s_g$ and $s_{x_{a,r}}=s_{x_{v_{a,r}}}$. For each equality constraint $s_p=s_q$, the construction of Lemma~\ref{lem:unique-lag-embed} allows the shared spin $s_z$ to be set to $-s_p=-s_q$; the two intermediate spins may be chosen arbitrarily. Thus every cardinality-two row has zero correlation. Since every source set is 2-2 split, Eq.~(\ref{masked-eq:clause-reduced}) also gives zero correlation for every cardinality-four clause row. Therefore the target instance has zero energy.

\emph{Soundness.} Suppose the target instance has zero energy. Every cardinality-two row introduced by Lemma~\ref{lem:unique-lag-embed} has zero correlation, so all equality constraints are satisfied. Consequently, for every source set $S_a$ and $r\in\dbn{1}{4}$, $s_{g_{a,r}}=s_g$ and $s_{x_{a,r}}=s_{x_{v_{a,r}}}$. Each clause row also has zero correlation, and Eq.~(\ref{masked-eq:clause-reduced}) therefore implies $\sum_{r=1}^4 s_{x_{v_{a,r}}}=0$. Thus exactly two representative spins in every source set equal $+1$ and exactly two equal $-1$. The assignment $v\mapsto s_{x_v}$ is therefore a valid 2-2 splitting of the source instance.

\emph{Polynomial size.} There is one cardinality-four row for each source set and two cardinality-two rows for each of the $8m$ equality constraints. Lemma~\ref{lem:unique-lag-embed} bounds the largest introduced spin position by a polynomial in $t+m$, so the unary encoding of $N$ has polynomial length. The mask contains only polynomially many explicitly listed pair terms, and all positions and lags are constructed using polynomially many arithmetic operations. Hence the target instance has polynomial size and can be constructed in polynomial time. This proves NP-hardness, and membership in $\NP$ gives NP-completeness.
\qed
\end{proof}

\section{Contiguous supports}\label{sec:contiguous}\label{ch:single-window}

A single-window Masked-LABS instance is a mask $M$ with the additional property that every nonempty support is one interval of consecutive start positions. For each active lag $k$, there are integers $a_k$ and $D_k\ge1$ satisfying
\begin{equation}
1\le a_k\le N-k-D_k+1
\end{equation}
such that
\begin{equation}
I_k^M=\dbn{a_k}{a_k+D_k-1}.
\end{equation}

The masked row is therefore
\begin{equation}
C_{k,a_k,D_k}(s)
:=\sum_{i=a_k}^{a_k+D_k-1}s_is_{i+k}
\end{equation}
where the \emph{window length} is $D_k=|I_k^M|$. 

Figure~\ref{fig:support-geometries}(d) illustrates this support geometry. The optimization and threshold decision problems are the restrictions of Masked-LABS to such masks, under the same canonical encoding as in Eq.~(\ref{eq:masked-encoding}). Since single-window instances form a subclass of Masked-LABS, Proposition~\ref{prop:masked-three} immediately gives polynomial-time zero-energy feasibility when every window has length at most three.

\subsection{Bond-copy constraints}

To construct equality constraints using length-two windows, we express consecutive correlation terms through nearest-neighbor bond variables. For a spin sequence $s\in\Sset^N$, define $b_j:=s_js_{j+1}$ for $j\in\dbn{1}{N-1}$. Every assignment $b\in\Sset^{N-1}$ is realized by a spin sequence: choose $s_1\in\Sset$ arbitrarily and set $s_{j+1}:=s_jb_j$ recursively. For lag $k$, write $y_i^{(k)}:=s_is_{i+k}$. Then
\begin{equation}\label{window-eq:ratio}
y_{i+1}^{(k)}=y_i^{(k)}b_i b_{i+k}.
\end{equation}

\begin{lemma}[Bond NOT]\label{lem:window-not}
For a length-two window of lag $k$ beginning at $a$, $C_{k,a,2}(s)=0$ if and only if $b_{a+k}=-b_a$. If $k\ge2$, its two selected lag-$k$ edges form a matching.
\end{lemma}
\begin{proof}
By Eq.~(\ref{window-eq:ratio}), $y_{a+1}^{(k)}=y_a^{(k)}b_ab_{a+k}$, and therefore $C_{k,a,2}(s)=y_a^{(k)}(1+b_ab_{a+k})$. Since $y_a^{(k)}\in\Sset$, the correlation vanishes if and only if $b_ab_{a+k}=-1$, equivalently $b_{a+k}=-b_a$. The selected spin edges are $\{a,a+k\}$ and $\{a+1,a+1+k\}$; when $k\ge2$, their four endpoints are distinct.
\qed
\end{proof}

Two Bond-NOT rows sharing an intermediate bond can enforce equality between two bond variables. Lemma~\ref{lem:window-bond-embed} next shows how these copy constraints can be implemented while preserving the required geometric restrictions.

\begin{lemma}[Unique-lag bond copy]\label{lem:window-bond-embed}
Let $P_0\subset\mathbb Z_{>0}$ be a nonempty finite set of bond positions, let $F_0\subset\mathbb Z_{>0}$ be a finite set of lags unavailable to the new windows, and let $\mathcal R=((p_1,q_1),\ldots,(p_R,q_R))\in(P_0^2)^R$ specify equality constraints $b_{p_r}=b_{q_r}$ with $p_r\ne q_r$ for every $r$. Write $D:=\max P_0$. There is a polynomial-time construction of distinct intermediate bond positions $z_1,\ldots,z_R$, all greater than $D$, such that each equality constraint $b_{p_r}=b_{q_r}$ is implemented by two length-two windows whose boundary-bond position pairs are $(p_r,z_r)$ and $(q_r,z_r)$, respectively, and such that:
\begin{enumerate}
\item the $2R$ new lags $z_r-p_r$ and $z_r-q_r$ are pairwise distinct and lie outside $F_0$;
\item every new lag is at least $2$, so the two selected edges in each new window form a matching;
\item each new lag supports exactly one contiguous length-two window; and
\item the largest intermediate position satisfies
\begin{equation}\label{window-eq:zbound}
\max_r z_r\le D+2R|F_0|+2R^2.
\end{equation}
\end{enumerate}
\end{lemma}

The proof is given in Appendix~\ref{app:single-window-embedding}.

\subsection{Length-four windows}

The length-four windows are used to encode constraints requiring two positive and two negative signs. We first establish the corresponding identity and then specify the window positions used in the reduction.

Consider a length-four window with start $a$ and lag $k$, and write $y_t:=s_{a+t}s_{a+k+t}$ for $t\in\dbn{0}{3}$. Suppose its six boundary bonds satisfy
\begin{equation}\label{window-eq:six}
(b_a,b_{a+1},b_{a+2};\,b_{a+k},b_{a+k+1},b_{a+k+2})
=(x_1,x_2,x_3;\,x_2,x_3,x_4).
\end{equation}
By Eq.~(\ref{window-eq:ratio}), $y_1=y_0x_1x_2$, $y_2=y_0x_1x_3$, and $y_3=y_0x_1x_4$. Hence
\begin{equation}\label{window-eq:four}
C_{k,a,4}(s)=y_0x_1(x_1+x_2+x_3+x_4).
\end{equation}

Thus the window has zero correlation exactly when $x_1,x_2,x_3,x_4$ contain two $+1$ signs and two $-1$ signs.

We use this equivalence to construct a reduction from the NP-hard all-split decision version of 2-2 Set Splitting \cite{Guruswami,SpielmanZhang2022}. Let the source instance have universe $U=\{v_1,\ldots,v_t\}$ and four-element sets $S_1,\ldots,S_m$, with $m\ge1$.

To realize the required boundary-bond configuration, assign the representative bond position $r_j:=j$ to each $v_j$.  For each $a\in\dbn{1}{m}$, define the clause lag $K_a:=a+3$, so $K_a\ge4$ and the clause lags are pairwise distinct. Let
\begin{equation*}
K_{\max}:=m+3,\quad W:=K_{\max}+4,\quad
A_a:=t+2+(a-1)W.
\end{equation*}

Write $S_a=(v_{a,1},v_{a,2},v_{a,3},v_{a,4})$ in an arbitrary fixed order. Define the six boundary-bond positions by
\begin{equation*}
\begin{aligned}
q_{a,1}&:=A_a,
&\quad q_{a,2}&:=A_a+1,
&\quad q_{a,3}&:=A_a+2,\\
q_{a,4}&:=A_a+K_a,
&\quad q_{a,5}&:=A_a+K_a+1,
&\quad q_{a,6}&:=A_a+K_a+2.
\end{aligned}
\end{equation*}

Associate these positions, respectively, with $v_{a,1},v_{a,2},v_{a,3},v_{a,2},v_{a,3},v_{a,4}$. Thus $v_{a,2}$ and $v_{a,3}$ each have two boundary-bond copies, as required by Eq.~(\ref{window-eq:six}).

\begin{lemma}[Clause geometry]\label{lem:window-clause-geom}
The bond positions and clause lags defined above satisfy the following properties.
\begin{enumerate}
\item The six boundary positions associated with each source set are pairwise distinct, and the corresponding position sets for distinct source sets are disjoint.
\item The length-four row at start $A_a$ and lag $K_a$ is a matching.
\item All clause lags are pairwise distinct.
\item The largest bond position defined before adding the copy windows is $O(t+m^2)$.
\end{enumerate}
\end{lemma}

\begin{proof}
For each $a$, $K_a\ge4$, so the boundary-bond intervals $\dbn{A_a}{A_a+2}$ and $\dbn{A_a+K_a}{A_a+K_a+2}$ are disjoint. Hence the six boundary positions associated with $S_a$ are pairwise distinct. All these positions lie between $A_a$ and $A_a+K_{\max}+2$, an interval of width $K_{\max}+2<W$. Since $A_{a+1}-A_a=W$, the boundary-position sets associated with distinct source sets are disjoint.

The length-four row associated with $S_a$ selects the four spin edges $\{A_a+j,A_a+K_a+j\}$ for $j\in\dbn{0}{3}$. Because $K_a\ge4$, their eight endpoints are distinct, so the row is a matching. The clause lags are pairwise distinct because $K_a=a+3$. Finally, the largest boundary-bond position is 
\begin{equation*}
\begin{aligned}
A_m+K_m+2
&= t + 2 + (m-1)(m+7) + (m+3) + 2\\
&=O(t+m^2).
\end{aligned}
\end{equation*}
\qed
\end{proof}

\subsection{Single-window NP-completeness}

\begin{theorem}[Window-length-four NP-completeness]\label{thm:single-window}
Zero-energy feasibility for single-window Masked-LABS is NP-complete when every nonempty window has length at most four. The problem remains NP-complete even under the simultaneous restrictions that every nonempty window has length exactly two or four, every active lag is used by exactly one window, and every selected lag row is a matching.
\end{theorem}
\begin{proof}
Membership in $\NP$ follows from the unary-length support-list encoding: a certificate is a spin sequence in $\Sset^N$, and all selected pair products can be evaluated in polynomial time.

\emph{Construction.}
For NP-hardness, take a 2-2 Set Splitting instance with universe $U=\{v_1,\ldots,v_t\}$ and four-element sets $S_1,\ldots,S_m$. Use the representative bond positions, clause lags, starts, and boundary-bond positions defined above, and include for each $S_a$ the length-four window at start $A_a$ and lag $K_a$.

Set $F_0:=\{K_1,\ldots,K_m\}$ to reserve the lags used by the length-four clause windows; the length-two copy windows must use lags outside $F_0$. Write $r(v_j):=r_j$ for the representative bond position of each source variable $v_j$. For every boundary-bond position $q_{a,i}$, impose an equality constraint between $b_{q_{a,i}}$ and the representative bond of its assigned variable, as specified above. Although $v_{a,2}$ and $v_{a,3}$ each label two boundary positions, every position requires a separate equality constraint. Thus each of the $m$ source sets contributes six constraints, giving $R=6m$.

Let $P_0$ consist of all representative and boundary-bond positions, and let $D:=\max P_0$. Apply Lemma~\ref{lem:window-bond-embed} to all $6m$ equality constraints. It introduces two length-two copy windows for each constraint, with pairwise distinct lags lying outside $F_0$.

After all copy windows have been added, let $Z$ be the largest bond position used, set $N:=Z+1$, and let $M$ be the mask consisting of all clause and copy windows. The target instance is $\enc_{\mathrm{Masked}}(N,M,0)$.

\emph{Polynomial size and geometry.}
By Lemma~\ref{lem:window-clause-geom}, $D=O(t+m^2)$. Since $R=6m$ and $|F_0|=m$, Lemma~\ref{lem:window-bond-embed} gives $Z\le D+2(6m)m+2(6m)^2=O(t+m^2)$. Since the target sequence length is $N=Z+1$, this bound also guarantees that its unary encoding has polynomial length. Moreover, the construction uses $m$ clause windows and $12m$ copy windows, each containing at most four pair terms. Thus the mask has polynomially many explicitly specified terms, and the complete target instance has polynomial encoding length.

Having established the polynomial size of the target instance, we verify that it also satisfies the required geometric restrictions. All clause lags are distinct, all copy lags are distinct, and the copy lags lie outside $F_0$; therefore every active lag supports exactly one window. Clause windows have length four and are matchings by Lemma~\ref{lem:window-clause-geom}; copy windows have length two and are matchings by Lemma~\ref{lem:window-not}.

\emph{Completeness.}
Suppose the source instance has a 2-2 splitting. Assign each representative bond its source sign and set every boundary-bond copy equal to its associated representative. For each equality constraint $b_p=b_q$, set the intermediate bond introduced by Lemma~\ref{lem:window-bond-embed} to $-b_p=-b_q$. Then both copy rows have zero correlation by Lemma~\ref{lem:window-not}. Assign all remaining bonds arbitrarily.

For every source set, the six boundary bonds then have the pattern in Eq.~(\ref{window-eq:six}), and the four representative signs are 2-2 split. Eq.~(\ref{window-eq:four}) therefore gives zero correlation for every clause row. Since every bond assignment is realizable by a spin sequence, the target instance has zero energy.

\emph{Soundness.}
Suppose the target instance has zero energy. Every length-two copy row has zero correlation. For each equality constraint, its two copy rows therefore force the boundary-bond copy to equal the corresponding representative. Thus the six boundary bonds of every source set have the pattern in Eq.~(\ref{window-eq:six}) for its four representative signs.

Each length-four clause row also has zero correlation, so Eq.~(\ref{window-eq:four}) implies that the four representative signs contain two $+1$ signs and two $-1$ signs. The representative bonds therefore define a valid 2-2 splitting of the source instance.

The construction is a polynomial-time many-one reduction. Hence the problem is NP-hard, and membership in $\NP$ gives NP-completeness.
\qed
\end{proof}

\section{Complete-lag selection}\label{sec:complete-lag}

\begin{definition}[Complete-lag selection]\label{def:complete-lag}
A complete-lag selection instance consists of a sequence length $N$ and an active-lag set $A\subseteq\dbn{1}{N-1}$. Each active lag retains its full LABS autocorrelation, so for $s\in\Sset^N$ the energy is
\begin{equation*}
E_A(s):=\sum_{k\in A}C_k(s)^2.
\end{equation*}
\end{definition}

Figure~\ref{fig:support-geometries}(e) illustrates this restriction to complete lag diagonals. The optimization problem minimizes $E_A(s)$ over $s\in\Sset^N$. A threshold instance is represented canonically by $\enc_{\mathrm{CLS}}(N,A,B)$, where $N$ is encoded in unary, $A$ as a canonical explicit list, and the nonnegative threshold $B$ in binary. The corresponding decision language is
\begin{equation}
L_{\mathrm{CLS}}
:=\{\enc_{\mathrm{CLS}}(N,A,B):\exists s\in\Sset^N,\ E_A(s)\le B\}.
\end{equation}

In Masked-LABS notation, complete-lag selection corresponds to the restriction $I_k^M=\dbn{1}{N-k}$ for $k\in A$ and $I_k^M=\varnothing$ otherwise. Standard LABS is recovered when every lag is active. This restriction also imposes a parity condition: each active correlation $C_k(s)$ contains $N-k$ terms and therefore cannot vanish when $N-k$ is odd. Consequently, zero-energy feasibility requires $N-k$ to be even for every active lag. In particular, when $N$ is odd, every active lag must be odd.

\subsection{Odd lengths}

For odd sequence lengths, the classical skew-symmetric construction of Golay \cite{Golay1972} provides binary sequences whose aperiodic autocorrelations vanish at every odd lag. Combined with the parity obstruction above, this result yields the following classification for complete-lag selection.

\begin{theorem}[Odd-length classification]\label{thm:odd-original}
Let $N=2h+1$. There exists a sequence $s\in\Sset^N$ such that $C_k(s)=0$ for every odd $k\in\dbn{1}{N-1}$. Consequently, a complete-lag-selection instance of odd length has zero energy if and only if every active lag is odd.
\end{theorem}
\begin{proof}
Necessity follows from the parity observation above: when $N$ is odd and $k$ is even, $N-k$ is odd, so $C_k(s)\ne0$.

For sufficiency, set $s_1=\cdots=s_{h+1}=1$ and define the remaining entries by $s_{N+1-i}:=(-1)^{h+i+1}s_i$ for $i\in\dbn{1}{h}$. Let $q_j:=s_js_{N+1-j}$. A direct check from the mirror rule gives $q_j=(-1)^{h+j+1}$ for every $j\in\dbn{1}{N}$.

Fix an odd lag $k$. For each term $T_i:=s_is_{i+k}$, pair it with the term beginning at $i':=N-k+1-i$. Since $N-k$ is even, this pairing has no fixed points. Moreover, $T_iT_{i'}=q_iq_{i+k}=(-1)^k=-1$. Thus every paired pair of terms has opposite signs, and all terms in $C_k(s)$ cancel. Hence $C_k(s)=0$ for every odd $k$.
\qed
\end{proof}

\subsection{Even lengths}

Let $N=2h$ and separate the odd and even spin positions by setting $u_j:=s_{2j-1}$ and $v_j:=s_{2j}$ for $j\in\dbn{1}{h}$. For $r\in\dbn{1}{h-1}$, splitting the start indices by parity gives
\begin{equation}\label{eq:interleave}
C_{2r}(s)=C_r(u)+C_r(v).
\end{equation}

For even $N$, the parity condition excludes every instance containing an active odd lag. It therefore remains only to analyze active even lags, for which Eq.~(\ref{eq:interleave}) gives the following exact reformulation.

\begin{definition}[SLCP]\label{def:slcp}
Let $h\ge1$ and $R\subseteq\dbn{1}{h-1}$. The \emph{selected-lag complementary-pair feasibility problem}, abbreviated $\SLCP(h,R)$, asks whether there exist $u,v\in\Sset^h$ such that $C_r(u)+C_r(v)=0$ for every $r\in R$.
\end{definition}

\begin{proposition}[Even-lag equivalence]\label{prop:slcp-equiv}
Let $N=2h$ and suppose the active-lag set is $A=\{2r:r\in R\}$ for some $R\subseteq\dbn{1}{h-1}$. The complete-lag instance has zero energy if and only if $\SLCP(h,R)$ is feasible.
\end{proposition}
\begin{proof}
By Eq.~(\ref{eq:interleave}), the condition $C_{2r}(s)=0$ for every $r\in R$ is exactly $C_r(u)+C_r(v)=0$ for the odd- and even-position subsequences $u$ and $v$. Conversely, any feasible pair $u,v$ can be interleaved uniquely by setting $s_{2j-1}=u_j$ and $s_{2j}=v_j$. Thus the two feasibility conditions are equivalent.
\qed
\end{proof}

\section{Conclusions}\label{sec:discussion}

This paper analyzes standard LABS, in which the sequence length fixes all correlation supports, and the input-dependent variants Pair-Set-LABS and Masked-LABS, in which the selected pair interactions are explicitly specified. For the unary standard-LABS decision language, membership in $\NP$ and the sparsity of its normalized form imply that Karp NP-hardness would force $\Ptime=\NP$. This consequence does not establish a polynomial-time algorithm for standard LABS, whose general decision complexity remains unresolved.

For the input-dependent support models, zero-energy feasibility exhibits a sharp complexity threshold. Odd-cardinality rows cannot have zero correlation, while the zero conditions for two-term rows reduce to affine equations over $\Ftwo$. Consequently, instances with at most three terms per group or row are polynomial-time decidable. In contrast, allowing four terms yields NP-completeness for both Pair-Set-LABS and Masked-LABS.

NP-completeness also persists for single-window Masked-LABS, where each active lag has a single contiguous support interval. The reduction remains valid using only length-two and length-four windows, with distinct active lags and matching rows. Thus the hardness does not depend on arbitrary pair groups or disconnected supports within a lag.

A different situation arises when every active lag retains its full LABS correlation. For odd sequence lengths, zero-energy feasibility admits an exact parity classification. For even lengths, active odd lags preclude zero energy, while the remaining case is equivalent to $\SLCP$. Its general complexity remains open. This identifies a boundary not resolved by the present reductions: the effect of requiring complete lag correlations rather than allowing the input to select individual terms within each lag.

{\small\rmfamily \noindent\textbf{Declarations}\quad During the preparation of this work, the author used Claude 3.5 Sonnet to improve English phrasing and readability. The author reviewed and edited the content as needed and takes full responsibility for the content of the publication.}

\vspace{6pt}
{\small\rmfamily \noindent\textbf{Data availability}\quad There exists no data associated with this manuscript.}

\vspace{6pt}
{\small\rmfamily \noindent\textbf{Conflict of Interest}\quad The author declares no conflict of interest.}

\appendix
\section*{Appendix}
\section{Distinct Positive 1-in-3-SAT}\label{app:distinct-p13}

%Schaefer \cite{Schaefer} considers the Boolean relation $R:\Bset^3\to\Bset$, where $R(x,y,z)=1$ if and only if exactly one of $x,y,z$ equals $1$, equivalently $x+y+z=1$ over the integers. In the associated satisfiability problem, each constraint is an application $R(x_i,x_j,x_\ell)$ of this relation, and the same variable may occur in more than one argument position. This repetition is problematic for the Pair-Set-LABS clause construction because two prescribed edges may then coincide, and duplicate edges are identified in the set $E_a$. 

%We therefore use a restricted decision problem in which every constraint $R(x_i,x_j,x_\ell)$ contains three pairwise distinct variables and no variable is negated.

The Boolean relation $R:\Bset^3\to\Bset$ is defined by $R(x,y,z)=1$ if and only if exactly one of $x,y,z$ equals $1$,
equivalently $x+y+z=1$ over the integers. In the associated satisfiability problem, each constraint has the form $R(x_i,x_j,x_\ell)$, and variables may be repeated among its argument positions. This problem belongs to the Boolean constraint framework studied by Schaefer \cite{Schaefer}.

Repeated variables are inconvenient for the Pair-Set-LABS clause construction, because two prescribed edges may then coincide and duplicate edges are identified in the set $E_a$. We therefore consider the restriction in which every constraint $R(x_i,x_j,x_\ell)$ contains three pairwise distinct variables and all variables occur positively.

\begin{definition}[Distinct Positive 1-in-3-SAT]\label{pair-src-def:distinct-p13}
An instance consists of a finite set $X$ of Boolean variables and a conjunction of constraints $R(x_i,x_j,x_\ell)$, where $x_i,x_j,x_\ell\in X$ are pairwise distinct. The question is whether there exists an assignment of the variables in $X$ to $\Bset$ such that every constraint evaluates to $1$.
\end{definition}

\begin{lemma}[Forcing gadgets]\label{pair-src-lem:forcing}
For pairwise distinct Boolean variables $x,p,q,r$, define
\begin{equation}
\begin{aligned}
F_0(x;p,q,r)
&:=R(x,p,q)\wedge R(x,p,r)\wedge R(p,q,r),\\
F_1(x;p,q,r)
&:=R(x,p,q)\wedge R(x,p,r)\wedge R(x,q,r).
\end{aligned}
\end{equation}

Every relation in these formulas contains three pairwise distinct variables. Moreover, $F_0(x;p,q,r)$ is satisfiable over $p,q,r$ if and only if $x=0$, while $F_1(x;p,q,r)$ is satisfiable over $p,q,r$ if and only if $x=1$. The respective satisfying assignments are uniquely $(x,p,q,r)=(0,1,0,0)$ and $(x,p,q,r)=(1,0,0,0)$.
\end{lemma}
\begin{proof}
For $F_0$, the first two constraints imply $q=r$. The third then gives $p+2q=1$, which forces $q=r=0$ and $p=1$. Substituting into either of the first two constraints gives $x=0$. Conversely, $(0,1,0,0)$ satisfies all three constraints.

For $F_1$, the first two constraints again imply $q=r$. The third gives $x+2q=1$, which forces $q=r=0$ and $x=1$. Substituting into either of the first two constraints gives $p=0$. Conversely, $(1,0,0,0)$ satisfies all three constraints.
\qed
\end{proof}

\begin{lemma}[NP completeness]\label{pair-src-lem:distinct-p13-npc}
Distinct Positive 1-in-3-SAT is NP-complete.
\end{lemma}
\begin{proof}
Membership in $\NP$ follows by using a Boolean assignment as a certificate. Each constraint has three argument positions and can be checked in constant time, so all constraints can be verified in polynomial time.

For NP-hardness, reduce from Schaefer's one-in-three satisfiability problem. Let $\Phi$ be an arbitrary instance. Construct an instance $\Phi'$ by replacing each constraint independently according to the number of distinct variables among its three argument positions. 

\emph{a) Three distinct variables.} If the three arguments are pairwise distinct, include the same constraint in $\Phi'$.

\emph{b) Exactly two distinct variables.}Since $R$ is symmetric in its arguments, such a constraint can be written as $R(x,x,y)$ with $x\ne y$. It is satisfied if and only if $2x+y=1$ over the integers, which occurs exactly when $x=0$ and $y=1$. Replace it by $F_0(x;p_0,q_0,r_0)\wedge F_1(y;p_1,q_1,r_1)$, using six auxiliary variables that are pairwise distinct, distinct from $x,y$, and used in no other replacement. By Lemma~\ref{pair-src-lem:forcing}, for fixed $x,y$ this conjunction can be satisfied by the auxiliary variables if and only if $x=0$ and $y=1$. Hence the replacement is equivalent to the original constraint.

\emph{c) One distinct variable.} Such a constraint has the form $R(x,x,x)$ and is unsatisfiable, since $3x$ is either $0$ or $3$, never $1$. Replace it by $F_0(x;p_0,q_0,r_0)\wedge F_1(x;p_1,q_1,r_1)$, using six auxiliary variables that are pairwise distinct, distinct from $x$, and used in no other replacement. By Lemma~\ref{pair-src-lem:forcing}, the first formula requires $x=0$ while the second requires $x=1$; therefore the replacement is also unsatisfiable.

Every constraint of $\Phi'$ contains three pairwise distinct unnegated variables. Moreover, the local replacements above preserve satisfiability. Because the auxiliary variables introduced for different constraints are disjoint, any satisfying assignment of $\Phi$ can be extended independently to a satisfying assignment of $\Phi'$. Conversely, restricting any satisfying assignment of $\Phi'$ to the original variables satisfies every constraint of $\Phi$. Hence $\Phi$ is satisfiable if and only if $\Phi'$ is satisfiable.

Each constraint of $\Phi$ produces at most six constraints and six auxiliary variables in $\Phi'$. Thus the number of variables and constraints increases by only a constant factor, and the transformation is computable in polynomial time. Therefore $\Phi\mapsto\Phi'$ is a polynomial-time many-one reduction. Since the source problem is NP-complete, Distinct Positive 1-in-3-SAT is NP-hard; together with membership in $\NP$, it is NP-complete.
\qed
\end{proof}

\section{Masked-LABS embedding}\label{app:masked-embedding}

\begin{proof}[Lemma~\ref{lem:unique-lag-embed}]
We implement the equality constraints sequentially, introducing two masked rows for each constraint. Initially, set $P:=P_0$ and $F:=F_0$. At each step, $P$ contains the original spin positions together with the auxiliary positions introduced so far, while $F$ contains the initially forbidden lags together with those assigned to earlier rows. New positions and lags must avoid these respective sets. For the next constraint $s_p=s_q$ with $p\ne q$, we seek an auxiliary position $z>\max P$ with the same parity as $p$ and $q$.

Once $z$ has been chosen, define the midpoint positions $m_p:=(p+z)/2$ and $m_q:=(q+z)/2$. Since $p,q,z$ have the same parity, both midpoints are integers and lie strictly between their respective endpoints. The first cardinality-two masked row consists of the edges $\{p,m_p\}$ and $\{m_p,z\}$, which share the lag $k_p:=(z-p)/2$. Similarly, the second row consists of $\{q,m_q\}$ and $\{m_q,z\}$, with lag $k_q:=(z-q)/2$. By Lemma~\ref{lem:spin-not-wire}, zero correlation in these rows gives $s_z=-s_p$ and $s_z=-s_q$, and therefore $s_p=s_q$.

Because $p\ne q$, the lags $k_p$ and $k_q$ are distinct. We choose $z$ so that $k_p,k_q\notin F$ and $m_p,m_q\notin P$. For each $f\in F$, a conflict $k_p=f$ or $k_q=f$ can occur only when
$z=p+2f$ or $z=q+2f$, so lag conflicts exclude at most $2|F|$ candidate values of $z$. Likewise, for each $u\in P$, a collision $m_p=u$ or $m_q=u$ can occur only when $z=2u-p$ or $z=2u-q$, so collisions with previously used spin positions exclude at most $2|P|$ candidates. Since $z>\max P$, the position $z$ itself is not already in $P$. Hence among $2|F|+2|P|+1$ successive candidates of the required parity above $\max P$, at least one is admissible.

Once an admissible $z$ has been found, update the sets of used spin positions and unavailable lags by setting $P\leftarrow P\cup\{z,m_p,m_q\}$ and $F\leftarrow F\cup\{k_p,k_q\}$. Each equality constraint introduces at most three spin positions and two lags. Consequently, after processing $r$ equality constraints, $|P|\le |P_0|+3r$ and $|F|\le |F_0|+2r$.

For each new constraint, candidate positions $z>\max P$ are examined in increasing order, with the required parity, until an admissible position is found. The number of candidates examined at each step is $O(|P_0|+|F_0|+r)$. Therefore, the increase in the largest used spin position at each step is polynomially bounded. Summing these increases over the $R$ equality constraints yields the polynomial bound on the largest spin position stated in the lemma. The total number of candidate tests is also polynomial, so the construction can be carried out in polynomial time.
\qed
\end{proof}

\begin{proof}[Lemma~\ref{lem:masked-clause-place}]
All spin positions associated with clause $a$ are even by construction. The four positions $g_{a,r}=B_a+2(r-1)$, $r\in\dbn{1}{4}$, are pairwise distinct, and so are the positions $x_{a,r}=g_{a,r}+K_a$. Since $K_a\ge10$ while $g_{a,4}-g_{a,1}=6$, every $x_{a,r}$ lies strictly to the right of every $g_{a,s}$. Hence the eight endpoints are distinct, and the four edges $\{g_{a,r},x_{a,r}\}$ form a matching.

The spin positions associated with clause $a$ lie between $B_a$ and $B_a+K_a+6$. Since $K_a+6<K_a+8\le K_{\max}+8=W/2<W$ and consecutive bases differ by $W$, the position sets associated with distinct clauses are disjoint. Moreover, $K_a=2(a+4)$, so the clause lags are pairwise distinct.

Finally, the largest spin position is $B_m+K_m+6$. Since $B_m=2(t+2)+mW=O(t+m^2)$ and $K_m+6=O(m)$, the largest spin position is $O(t+m^2)$.
\qed
\end{proof}

\section{Single-window embedding}\label{app:single-window-embedding}

\begin{proof}[Lemma~\ref{lem:window-bond-embed}]
Set $F:=F_0$ and let $Z:=D$. Before each equality constraint is processed, $F$ contains the lags unavailable to subsequent windows, and $Z$ denotes the largest bond position assigned by the preceding steps (initially $Z=D$).

Consider an equality constraint $b_p=b_q$ with $p\ne q$. Choose an intermediate bond position $z>Z$ and use two length-two windows whose boundary-bond position pairs are $(p,z)$ and $(q,z)$. Their lags are $z-p$ and $z-q$, which are distinct because $p\ne q$.

To ensure that neither new lag belongs to $F$, examine candidate values of $z$ beginning at $Z+2$. For each $f\in F$, a conflict can occur only when $z=p+f$ or $z=q+f$, so at most $2|F|$ candidates are excluded. Hence among $2|F|+1$ consecutive candidates, at least one is admissible. Because $z\ge Z+2$ and $p,q\le D\le Z$, both new lags are at least $2$, and the corresponding rows are matchings by Lemma~\ref{lem:window-not}.

Choose the first admissible $z$, add the lags $z-p$ and $z-q$ to $F$, and replace $Z$ by $z$. By Lemma~\ref{lem:window-not}, the two rows have zero correlation exactly when $b_z=-b_p$ and $b_z=-b_q$. The two equalities are simultaneously satisfied if and only if $b_p=b_q$, so the two windows implement the required copy constraint.

Before processing the $r$th equality constraint, $|F|=|F_0|+2(r-1)$. Since at most $2|F|+1$ candidates are examined starting from $Z+2$, the increase in $Z$ is at most $2|F_0|+4(r-1)+2$. Summing these increases over all $R$ constraints gives $\max_r z_r\le D+2R|F_0|+2R^2$, which is the bound stated in the lemma. The same estimate shows that only polynomially many candidate positions are examined, so the construction is computable in polynomial time.
\qed
\end{proof}

\bibliographystyle{spmpsci}
\bibliography{references}
\end{document}